\documentclass[a4paper]{llncs}

\usepackage{color}  
\usepackage{makeidx}  

\usepackage[utf8]{inputenc}
\usepackage[T1]{fontenc}
\usepackage{amsfonts}
\usepackage{amssymb}
\usepackage{graphicx}
\usepackage[polish,english]{babel}
\newtheorem{fact}{Fact}
\begin{document}

\title{Hierarchy of persistence \\with respect to the length of action's disability}

\author{Kamila Barylska\thanks{The study is founded 
by the Polish National Science Center
(grant No.2013/09/D/ST6/03928).}\inst{1}\and Edward Ochma\'{n}ski\inst{2,3}\\
\email\{{khama,edoch\}{@}mat.umk.pl}}

\institute{Faculty of Mathematics and Computer Science, Nicolaus Copernicus University,\\ Chopina 12/18, 87-100 Toru\'{n}, Poland
\and Faculty of Mathematics and Computer Science, Nicolaus Copernicus University,\\ Chopina 12/18, 87-100 Toru\'{n}, Poland - retired
\and Institute of Computer Science, Polish Academy of Sciences,\\Jana Kazimierza 5, 01-248 Warszawa, Poland - retired
}

\maketitle

\begin{abstract}
The notion of persistence, based on the rule "no action can disable another one" is one of the classical notions in concurrency theory. It is also one of the issues discussed  in the Petri net theory. We recall two ways of generalization of this notion: the first is "no action can kill another one" (called l/l-persistence) and the second "no action can kill another enabled one" (called the delayed persistence, or shortly e/l-persistence). Afterwards we introduce a more precise notion, called e/l-k-persistence, in which one action disables another one for no longer than a specified number k of single sequential steps. Then we consider an infinite hie\-rarchy of such e/l-k persistencies. We prove that if an action is disabled, and not killed, by another one, it can not be postponed indefinitely.
Afterwards, we investigate the set of markings in which two actions are enabled simultaneously, and also the set of reachable markings with that feature. We show that the minimum of the latter is finite and effectively computable.
Finally we deal with decision problems about e/l-k persistencies. We show that all the kinds of e/l-k persistencies are decidable with respect to steps, markings and nets.


\end{abstract}

\section{Introduction}

Petri nets constitute a very useful and suitable tool for concurrent systems modeling. Thanks to them, we can not only model real systems, but also analyze their properties and design systems which fulfill given criteria. For many years, concurrent systems have been examined in the context of their compliance with certain desirable properties, which fits in with the trend of the so-called ethics of concurrent computations. One of the commonly found undesirable properties of concurrent systems is the presence of conflicts, and thus, one of the most desirable properties of them is conflict-freeness. The notion of persistence, proposed by Karp/Miller \cite{KarpMiller} is one of the the most important notions in concurrency theory. It is based on the behaviourally oriented rule "no action can disable another one", and generalizes the structurally defined notion of conflict-freeness.
\\ \\
The notion of persistence is one of the issues frequently discussed in the Petri net theory - \cite{BestDar,Grabowski,Hack,LandRob,Mayr,Koutny} and many others. It is being studied not only in terms of theoretical properties, and also as a useful tool for designing and analyzing software environments \cite{BestEsparza}. In engineering, persistence is a~highly desirable property, especially in a case of designing systems to be implemented in hardware. Many systems can not work properly without satisfying this property.
\\ \\
We say that an action of a processing system is persistent if, whenever it becomes enabled, it remains enabled until executed. A system is said to be persistent if all its actions are persistent. This classical notion has been introduced by Karp/Miller \cite{KarpMiller}. In section~2.6, we show two generalizations of the classical notion (defined in \cite{BarOch}): l/l-persistence and e/l-persistence which form the following hierarchy: $P_{e/e}\subseteq P_{l/l}\subseteq P_{e/l}$. An action is said to be l/l-persistent if it remains live until executed, and is e/l-persistent if, whenever it is enabled, it cannot be killed by another action. For uniformity, we name the traditional persistence notion e/e-persistence. Next, we recall that those kinds of persistence are decidable in place/transition nets.
\\ \\
In section 3.1, we extend the hierarchy mentioned above with an infinite hierarchy of e/l-persistent steps. A step $MaM'$ is said to be \emph{e/l-k-persistent} for some k~$\in \mathbb{N}$ if the execution of an action $a$ pushes the execution of any other enabled action away for at most k steps (more precise: if the execution of an action $a$ stops the enabledness of any other action, then the enabledness is restored not later than after k steps). 
\\ \\
In section 3.2 we study decision problems related to the notion of e/l-k-persistence. These problems include EL-k Step Persistence Problem and EL-k Marking Persistence Problem. We show that both problems are decidable (Theorem~\ref{t45} and Theorem \ref{t46}). 
\\ \\
The next problem we want to focus on is EL-k Net Persistence Problem. In order to check the decidability of the problem we need to take advantage of additional tools and facts.
That is why we investigate the set of markings in which two actions are enabled simultaneously, and also the set of reachable markings with that feature. We show that the minimum of the latter is finite and effectively computable. We also prove that if some action pushes the enabledness of another one away for more than k steps, then it also needs to happen in some minimal reachable marking enabling these two actions. In our proofs we use the decidability of the Set Reachability Problem (from \cite{BarOch}) and also we make use of the theory of residual sets of Valk/Jantzen \cite{ValkJantzen}. Finally, we show that e/l-k-persistence is decidable with respect to nets (Theorem \ref{t418}).

\mbox{ }\\
We also prove (section 3.3) that if an action of an arbitrary p/t-net is disabled (but not killed) by another one, it can not be postponed indefinitely. We show that if a p/t-net is e/l-persistent, then it is e/l-k-persistent for some k $\in \mathbb{N}$ (Theorem~\ref{t410}), and such a number k can be effectively found (Theorem~\ref{t420}). We also point, that the above-cited result does not hold for nets which do not have the monotonicity property (i.e. it is not true that the action enabled in some marking $M$ is also enabled in any marking $M'$ greater than $M$), for example for inhibitor nets.
\\ \\
The concluding section contains some questions and plans for further investigations.
\\ \\
A preliminary version of the paper was presented on the International Workshop on Petri Nets and Software Engineering (Hamburg, Germany, June 25-26, 2012) with electronical proceedings available online at CEUR-WS.org as Volume 851. The present paper is an improved and extended version of it.

\section{Basic Notions}

\subsection{Denotations}
The set of non-negative integers is denoted by $\mathbb{N}$. Given a set $X$, the cardinality (number of elements) of $X$ is denoted by $|X|$, the powerset (set of all subsets) by $2^X$, the cardinality of the powerset is $2^{|X|}$. Multisets over $X$ are members of  $\mathbb{N}^X$, i.e. functions from $X$ into $\mathbb{N}$.

\subsection{Petri Nets and Their Computations}

The definitions concerning Petri nets are mostly based on \cite{DeselReisig}.

\begin{definition}[Nets]
\label{d21}
\emph{Net} is a triple $\mathrm{N}=(P,T,F)$, where:
\begin{itemize}
\item $P$ and $T$ are finite disjoint sets, of \emph{places} and \emph{transitions}, respectively;
\item $F\subseteq P\times T \cup T \times P$ is a relation, called the \emph{flow relation}.
\end{itemize}
\end{definition}\mbox{ }\\
\indent For all $a\in T$ we denote:\\
\indent$^\bullet a = \{p \in P \mid (p,a) \in F\}$  -  the set of entries to $a$\\
\indent$a^\bullet = \{p \in P \mid (a,p) \in F\}$  -  the set of exits from $a$.\\
\\ \\
\newpage
Petri nets admit a natural graphical representation. Nodes represent places and transitions, arcs represent the flow relation. Places are indicated by circles, and transitions by boxes.
\\ \\
The set of all finite strings of transitions is denoted by $T^*$, the length of $w\in T^*$ is denoted by $|w|$, number of occurrences of a transition $a$ in a string $w$ is denoted by $|w|_a$, two strings $u,v \in T^*$ such that $(\forall a\in T)$ $|u|_a=|v|_a$ are said to be \emph{Parikh equivalent}, which is denoted by $u\equiv v$.
\\
\begin{definition}[Place/Transition Nets]
\label{d22}
\emph{Place/transition net} (shortly, \emph{p/t-net}) is a quadruple $\mathrm{S}=(P,T,F,M_0)$, where:
\begin{itemize}
\item $\mathrm{N}=(P,T,F)$ is a net, as defined above;
\item $M_0 \in \mathbb{N}^\mathrm{P}$ is a multiset of places, named the \emph{initial marking}; it is marked by \emph{tokens} inside the circles, capacity of places is unlimited.
\end{itemize}
\end{definition}\mbox{ }\\
Multisets of places are named \emph{markings}. In the context of p/t-nets, they are mostly represented by nonnegative integer vectors of dimension $|P|$, assuming that $P$ is strictly ordered. The natural generalizations, for vectors, of arithmetic operations $+$ and $-$, as well as the partial oder $ \leqslant$, all defined componentwise, are well known and their formal definitions are omitted here.
\\ \\
In this context, by $^\bullet a (a^\bullet)$ we understand a vector of dimension $|P|$ which has 1 in every coordinate corresponding to a place that is an entry to (an exit from, respectively) $a$ and 0 in other coordinates.
\\ \\
A transition $a \in T$ is \emph{enabled} in a marking $M$ whenever $^\bullet a\leq M$ (all its entries are marked). If $a$ is enabled in $M$, then it can be executed, but the execution is not forced. The execution of a transition $a$ changes the current marking $M$ to the new marking $M'=(M-^\bullet a)+a^\bullet $ (tokens are removed from entries, then put to exits). The execution of an action $a$ in a marking $M$ we call a (sequential) step. We shall denote $Ma$ for the predicate "a is enabled in $M$" and $MaM'$ for the predicate "$a$ is enabled in $M$ and $M'$ is the resulting marking".
\\ \\
This notions and predicates we extend, in a natural way, to strings of transitions: $M\varepsilon M$ for any marking $M$, and $MvaM''$  ($v \in T^*, a \in T$) iff $MvM'$ and $M'aM''$ .
\\ \\
\textbf{Remark:} Wherever this will not lead to confusion, we will use a notation $M'=Ma$ to denote the fact that the action $a$ is enabled in a marking $M$ and a marking $M'$ is the result of the execution of action $a$ in a marking $M$.
\\ \\
If $MwM'$, for some $w \in T^*$, then $M'$ is said to be \emph{reachable from $M$}; the set of all markings reachable from $M$ is denoted by $[M\rangle$ . Given a p/t-net $\mathrm{S}=(P,T,F,M_0)$, the set $[M_0\rangle$ of markings reachable from the initial marking $M_0$ is called the \emph{reachability set} of S, and markings in $[M_0\rangle$  are said to be \emph{reachable} in S.

\mbox{ }\\
A transition $a \in T$ is said to be \emph{live in a marking $M$} if there is a string $u\in T^*$ such that $ua$ is enabled in $M$. A transition $a \in T$ that is not live in a~marking $M$ is said to be \emph{dead in a marking $M$}. Let $M \in [M_0\rangle$ be a marking such that $MaM'$ for some $a\in T$, then if a transition $b \neq a$ is enabled (live) in $M$ and not enabled (not live) in $M'$, we say that (the execution of) $a$ \emph{disables} (\emph{kills}) $b$ in a marking $M$. We also say that an action $a$  \emph{disables} (\emph{kills}) $b$ (in a net $S$) if $a$ \emph{disables} (\emph{kills}) $b$ in some reachable marking $M$.
\\ 
\begin{definition} [Inhibitor nets ]
\label{d23}
\emph{Inhibitor net}  is a quintuple $\mathrm{S} = (P,T,F,I,M_0)$, where:
\begin{itemize}
\item $(P,T,F,M_0)$ is a p/t-net, as defined above;
\item $I\subseteq P\times T$ is the set of inhibitor arcs (depicted by edges ended with a small empty circle). Sets of entries and exits are denoted by $^\bullet a$ and $a^\bullet$, as in p/t-nets; the set of \emph{inhibitor entries} to $a$ is denoted by $^\circ a=\{p\in P \mid (p,a)\in I\}$.
\end{itemize}
\end{definition}
A transition $a \in T$ (of an inhibitor net) is enabled in a marking $M$ whenever $\bullet a\leq M$ (all its entries are marked) and $(\forall p\in{ ^\circ a})$ $M(p)= 0$ - all inhibitor entries to $a$ are empty. The execution of $a$ leads to the resulting marking $M'= (M-^\bullet a)+a^\bullet$.
\\ \\
The following well-known fact follows easily from Definitions \ref{d21} and \ref{d22}.
\begin{fact} [Diamond and big diamond properties]
\label{f24}
Any place/transition net possesses the following property:

\emph{Big Diamond Property}:

If $MuM' \ \& \ MvM'' \ \& \ u\approx v$ (Parikh equivalence), then $M' =M''$.

Its special case with $|u|=|v|=2$ is called the \emph{Diamond Property}:

If $MabM'\ \& \ MbaM''$, then $M' =M''$.
\end{fact}

\begin{definition}
\label{d261}
We say that a Petri net $\mathrm{S}=(P,T,F,M_0)$ has the \emph{monotonicity property} if and only if $(\forall w \in T^*)(\forall M,M' \in \mathbb{N}^{P}) \ Mw \land M\leq M' \Rightarrow M'w$.
\end{definition}

\begin{fact}
\label{f27}
P/t-nets have the monotonicity property.
\end{fact}

\begin{proof}
Obvious, since in p/t-nets the tokens of $M'-M$ can be regarded as frozen (disactive) tokens. 	
\end{proof}
	
\begin{fact}
\label{f28}
Inhibitor nets do not have the monotonicity property.
\end{fact}

\begin{proof}
Let us look at the example of Fig. \ref{Fig1}. It can be easily seen that $M_1<M_1'$. $M_1a$ holds but $M_1'a$ doesn't hold.	
\end{proof}

\begin{figure}[h]
\centering
\includegraphics[width=0.2\textwidth]{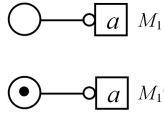}
\caption{Non-monotonic inhibitor net}
\label{Fig1}
\end{figure}

\subsection{Monoid $\mathbb{N}^\mathrm{k}$}
\begin{definition}[Monoid $\mathbb{N}^\mathrm{k}$, rational operations, rational subsets]
\label{def_mon}\\
The monoid $\mathbb{N}^\mathrm{k}$ is the set of $\mathrm{k}$-dimensional non-negative integer vectors with the componentwise addition $+$.
\\
If $X,Y\subseteq \mathbb{N}^\mathrm{k}$ then $X+Y\ =\ \{x+y\ |\ x\in X, y\in Y\}$ and the star operation is defined as $X^*=\bigcup\{X_i | i\in \mathbb{N}\}$, where $X_0=(0,\ldots,0)$ and $X_{i+1}= X_i+X$. The partial order $\leq$ is understood componentwise, and $<$ means $\leq$  and $\neq$ . Rational subsets of $\mathbb{N}^\mathrm{k}$ are subsets built from finite subsets with finitely many operations of union $\cup$ , addition $+$ and star $^*$.
\end{definition}
\begin{theorem}[Ginsburg/Spanier\cite{GinsburgSpanier}]
\label{twGins}
Rational subsets of $\mathbb{N}^\mathrm{k}$ form an effective boolean algebra (i.e. are closed under union, intersection and difference).

\end{theorem}
\begin{definition} [$\omega$-extension]
\label{d25}
Let $\mathbb{N}_\omega=\mathbb{N}\cup \{\omega \}$, where $\omega$  is a new symbol (denoting infinity). We extend, in a natural way, arithmetic operations: \\$\omega+\mathrm{n}=\omega$, $\omega-\mathrm{n}=\omega$, and the order: $(\forall \mathrm{n}\in \mathbb{N}) \ \mathrm{n}<\omega$. \\The set of k-dimensional vectors over $\mathbb{N}_\omega$ we shall denote by $\mathbb{N}_\omega^\mathrm{k}$, and its elements we shall call \emph{$\omega$-vectors}. Operations $+,-$ and the order $\leq$ in $\mathbb{N}_\omega^\mathrm{k}$ are componentwise.
\\ 
For $X\subseteq \mathbb{N}_\omega^\mathrm{k}$, we denote by Min($X$) the set of all minimal (wrt $\leq$) members of $X$, and by Max($X$) the set of all maximal (wrt $\leq$) members of $X$. Let $v,v' \in  \mathbb{N}_\omega^\mathrm{k}$ be $\omega$-vectors such that $v \leq v'$, then we say that \emph{$v'$ covers $v$} (\emph{$v$ is covered by $v'$}) .
\end{definition}
Let us recall the well known important fact known as the Dickson's Lemma.

\begin{lemma} [\cite{Dickson}]
\label{l26}
Any subset of incomparable elements of  $\mathbb{N}_\omega^\mathrm{k}$ is finite.
\end{lemma}

\begin{definition}[Closures, convex sets]
\label{def_clos}
\begin{itemize}
\item Let $x \in \mathbb{N}_\omega^\mathrm{k}$ and $X\subseteq\mathbb{N}_\omega^\mathrm{k}$. We denote:  $\downarrow x = \{z\in \mathbb{N}^\mathrm{k}\ |\ z\leq x\}$, \\$x\uparrow  = \{z\in \mathbb{N}^\mathrm{k}\ | \ x\leq z\}$,  $\downarrow X = \bigcup\{\downarrow \!x\ | \ x\in X\}$, $ X\uparrow = \bigcup\{ x\uparrow\ |\ x \in X\}$, and call the sets $\mathrm{left}$ and $\mathrm{right\ closures}$ of $X$, respectively;
\item A set $X\subseteq \mathbb{N}^\mathrm{k}$ such that $X= \downarrow X$ ($X=X\uparrow$) is said to be $\mathrm{left}$-($\mathrm{right}$-) $\mathrm{closed}$;
\item A set $X\subseteq \mathbb{N}^\mathrm{k}$ such that $X= \downarrow X=X\uparrow$  is said to be $\mathrm{convex}$.
\end{itemize}

\end{definition}
We also recall a fact proved in \cite{BarOch}:

\begin{proposition}
\label{p26}
Any convex subset of $\mathbb{N}^\mathrm{k}$ is rational.
\end{proposition}

\subsection{Reachability graph/tree and coverability graph}

Let us recall the notions of  \emph{a reachability graph/tree} and \emph{a coverability graph}. Their definitions can be also found in any monograph or survey about Petri nets (see \cite{DeselReisig,Starke} or arbitrary else). Reachability graphs/trees are used for studying complete behaviors of nets, but they are usually infinite, whichmakes an accurate analysis of them difficult. That is why we study coverability graphs, which represent the behaviours of nets only partially, but are always finite.
\\ \\
\emph{The reachability graph} of a p/t-net $\mathrm{S}=(P,T,F,M_0)$ is a couple $\mathrm{RG}=(\mathrm{G},M_0)$ where $[M_0\rangle\times P \times [M_0\rangle\supseteq\mathrm{G}=\{(M,a,M')\ |\ M\in[M_0\rangle\land MaM'\}$.
\\ \\
The reachability graph $\mathrm{RG}$ represents graphically the behaviour of the net $\mathrm{S}$. Vertices of the graph are reachable markings from the set $[M_0\rangle$, while edges are ordered pairs of reachable states, labeled by actions. More precisely: the edge $(M,a,M')\in\mathrm{G}$ iff $M$ is a state reachable from the initial marking $M_0$, an action $a\in T$ (the label of the edge $(M,a,M')$) is enabled in a state $M$ and $M'=Ma$. The existence of an edge $(M,a,M')$ in the reachability graph of the net $\mathrm{S}$ indicates that the marking $M$ is reachable in $\mathrm{S}$, the action $a$ is enabled in $M$ and after the execution ot the action $a$ in the marking $M$, the net $\mathrm{S}$ reaches the state $M'$.
\\ \\
Sometimes it is more convenient to use a special graph structure for listing all reachable markings of a given p/t-net, namely a tree structure. Such a tree is called \emph{a~reachability tree}. 
\\ \\For a given net $\mathrm{S}=(P,T,F,M_0)$ we construct its reachability tree $\mathrm{RT}$ proceeding as follows:
\begin{itemize}
\item We start with the initial marking $M_0$ which is the root vertex of the reachability tree.
\item For each action $a$ enabled in the initial marking of the net, we create a new vertex $M'$, such that $M'=M_0a$, and an edge $(M_0,a,M')$ leading from $M_0$ to $M$ labelled by $a$. 
\item We repeat the procedure for all the newly created vertices (markings).
\end{itemize}\mbox{ }\\
\textbf{Remark:} The construction of a reachability tree is a process potentially endless, as the structure is infinite in many cases.\\
\begin{definition}
Let $\mathrm{RT}$ be a reachability tree of a net $\mathrm{S}=(P,T,F,M_0)$. The \emph{k-component} of the reachability tree $\mathrm{RT}$ is the initial part of the tree of the depth~k (all vertices at depth lower than or equal to k).
\end{definition}\mbox{ }\\
In the case of a coverability tree it is convenient to present a constructional definition. That is why we introduce:
\\ \\
\textbf{Algorithm of the construction of a coverability graph}
\\ \\
We create a coverability graph for a p/t-net $\mathrm{S}=(P,T,F,M_0)$
\begin{itemize}
\item Step 0. \textit{An initial vertex}\\
We set $M_0$ \textbf{blue} for a start.\\
GOTO Step 1.\\ 
\item Step 1. \textit{Generating of new working vertices}\\
If there is no \textbf{blue} vertices then STOP.\\
We take an arbitrary \textbf{blue} vertex $M$ and draw from it all the arcs of the form $(M,t,M')$ for all $t\in T$ enabled in $M$, where $M'=Mt$. If the vertex $M'$ already exists (in any colour), then the newly created arc leads to the existing vertex (we do not create a new one). New vertices are set \textbf{yellow}. After drawing all such arcs we set the vertex $M$  \textbf{grey} (\textit{a final node}).\\
GOTO Step 2.\\ 
\item Step 2. \textit{Coverability checking}\\
If there is no \textbf{yellow} vertices GOTO Step 1.\\
We take an arbitrary \textbf{yellow} vertex $M$ and check for any of the paths from $M_0$ to $M$ whether a vertex $M'$ such that $M'\leq M$ lies on the path. If such a vertex exists then every coordinate of the marking $M$ greater than the corresponding coordinate of the marking $M'$ changes to $\omega$. Finally we set the vertex $M$ \textbf{blue}.\\
GOTO Step 2.\\ 
\end{itemize}
\begin{example}
\begin{figure}[h]
\centering
\includegraphics[width=1\textwidth]{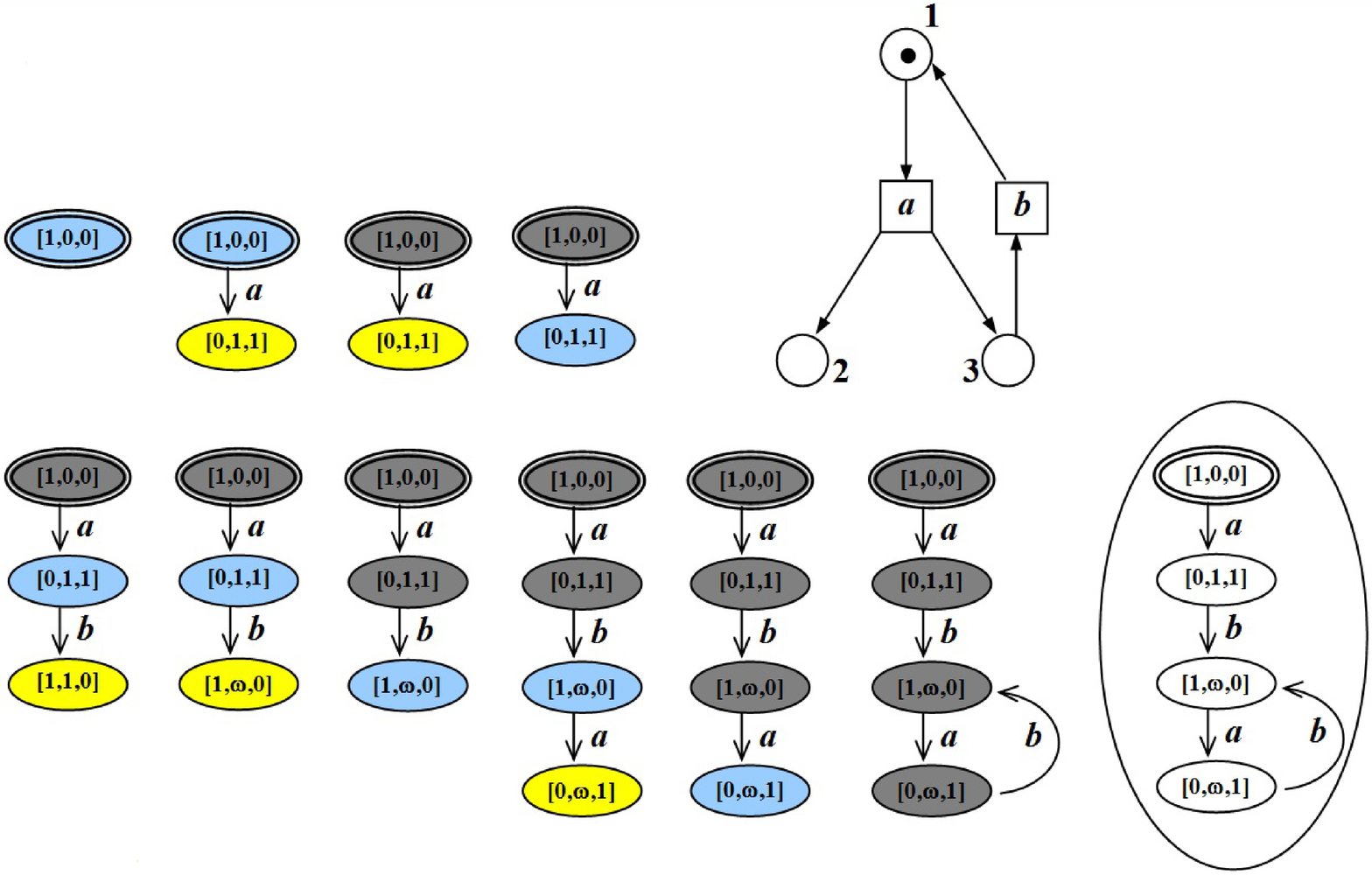}
\caption{A p/t-net and its coverability graph}
\label{FigC}
\end{figure}
\label{e_cov}
Let us look at the Example of Figure \ref{FigC}. A p/t-net and stages of the construction of its coverabilty graph are presented there.
\end{example}
\mbox{ }\\
\textbf{Remark:} A coverability graph is always finite. The proof is based on two facts: the Dickson's Lemma (Lemma~\ref{l26}) and the monotonicity property (Fact \ref{f27}).
\\ 
\subsection{Reachability and Coverability Problems}

Let us now recall very famous decision problems concerning Petri nets, namely the Reachability Problem and the Coverability Problem.
\\ \\
\textbf{Reachability Problem}
\\
\indent\textbf{\emph{Instance:}} P/t-net $\mathrm{S}=(\mathrm{N}, M_0)$, and a marking $M$.\\
\indent\textbf{\emph{Question:}} Is $M$ reachable in S?
\\ \\
\textbf{Coverability Problem}
\\
\indent\textbf{\emph{Instance:}} P/t-net $\mathrm{S}=(\mathrm{N}, M_0)$, and a marking $M$.\\
\indent\textbf{\emph{Question:}} Is $M$ coverable in S?
\\ \\
\textbf{Remark:} It is well known that the above problems are decidable (coverability: Karp/Miller \cite{KarpMiller}, Hack \cite{Hack}; reachability: Mayr \cite{Mayr}, Kosaraju \cite{Kosaraju}).

\subsection{Three Kinds of Persistence}

The notion of persistence is one of the classical notions in concurrency theory. The notion is recalled in \cite{BarOch} (named in the sequel e/e-persistence). Some of its generalizations: l/l-persistence and e/l-persistence are also introduced there.
\\ \\
\textbf{Note on terminology}\\
The notion of persistence in its classical meaning is a property of nets. The definition of \cite{LandRob} involves the entire concurrent system. 
If we choose to define the concept of persistence starting from actions by markings, ending with whole nets, the classic definition can be interpreted in two ways. Namely, one can consider concepts of \emph{persistence} and \emph{nonviolence}. An extensive discussion on the links between persistence and nonviolence can be found in \cite{Koutny}. In the context of \cite{BarMikOch} and  \cite{Koutny} it seems that it would be more appropriate to use the notion of nonviolence instead of using the concept of persistence. However, because our paper is an extension of \cite{BarOch}, we decided to stick to the concept of persistence.

\newpage
Let us sketch the notions of e/e-persistence, l/l-persistence and e/l-persistence informally. The classical e/e-persistence means "no action can disable another one", the l/l-persistence means "no action can kill another one" and the e/l-persistence means "no action can kill another enabled one". Let us go on to formal definitions.

\begin{definition}[Three kinds of persistence]
\label{d31}
Let $\mathrm{S}=(P,T,F,M_0)$ be a place/transition net. \\
If  $(\forall M \in [M_0\rangle)\ (\forall a,b \in T)$
\begin{itemize}
\item $Ma \land Mb \land a\neq b \Rightarrow Mab$, then $\mathrm{S}$ is said to be \emph{e/e-persistent};
\item $Ma \land (\exists u) Mub \land a\neq b \Rightarrow (\exists v) Mavb$, then $\mathrm{S}$ is said to be \emph{l/l-persistent};
\item $Ma \land Mb \land a\neq b \Rightarrow (\exists v) Mavb$, then $\mathrm{S}$ is said to be \emph{e/l-persistent}.
\end{itemize}\mbox{}\\
The classes of e/e-persistent (l/l-persistent, e/l-persistent) p/t-nets will be denoted by $\mathrm{P}_{\mathrm{e/e}}$, $\mathrm{P}_{\mathrm{l/l}}$ and $\mathrm{P}_{\mathrm{e/l}}$, respectively.
\end{definition}
In \cite{BarOch} one can find a proof of the following theorem:
\begin{theorem}\mbox{}\\
\label{t_hier}
The three classes of persistent place/transition nets
form an increasing hierarchy: $P_{e/e}\subseteq P_{l/l}\subseteq P_{e/l}$.
\end{theorem}
\begin{figure}[h]
\centering
\includegraphics[width=0.1\textwidth]{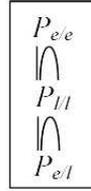}
\caption{A hierarchy of persistent nets}
\label{FigA}
\end{figure}
\begin{example}
\label{ex_fig_e_d}
To see the strictness of the above inclusion, let us look at the Examples of Figure \ref{Fig_E} and \ref{Fig_D} (derived from \cite{BarOch}).
\end{example}
\begin{figure}[h]
\centering
\includegraphics[width=0.2\textwidth]{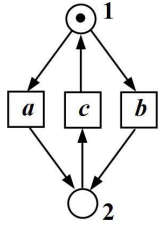}
\caption{The net is live, so $P_{\mathrm{l/l}}$, but not $P_{\mathrm{e/e}}$ }
\label{Fig_E}
\end{figure}
\begin{figure}[h]
\centering
\includegraphics[width=0.7\textwidth]{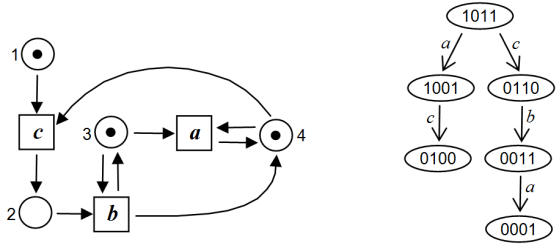}
\caption{The transition $a$ kills $b$ undirectly (because in the current marking, the transition $b$ is not enabled), so the net is $P_{\mathrm{e/l}}$, but not $P_\mathrm{{l/l}}$ }
\label{Fig_D}
\end{figure}
It is also shown in \cite{BarOch} that the following decision problems are decidable:
\\
\\
\textbf{\emph{Instance:}} A p/t-net $(\mathrm{N}, M_0)$ \\
\textbf{\emph{Questions:}}\\
\indent\textbf{EE Net Persistence Problem:} Is the net S e/e-persistent?\\
\indent\textbf{LL Net Persistence Problem:} Is the net S l/l-persistent?\\
\indent\textbf{EL Net Persistence Problem:} Is the net S e/l-persistent?\\
\\ 
\newpage
The proofs of decidability of the above problems need to put into work a very efficient result of Valk/Jantzen \cite{ValkJantzen} and benefit from the decidability of reachability problem (more specifically - decidability of the Set Reachability Problem for rational convex sets).
\\
The alternative proof of Theorem \ref{t418} uses exactly the same proving technique as the proofs of decidability of the persistence problems mentioned above.

\section{Properties of e/l-persistence}

\subsection{Hierarchy of e/l-persistence}

In the previous section we defined three kinds of persistence. Now, we extend the hierarchy mentioned above with an infinite hierarchy of e/l-persistent steps.

\begin{definition}[E/l-persistent steps - an infinite hierarchy]
\label{d41}\\
Let $\mathrm{S}=(P,T,F,M_0)$ be a p/t-net, let $M$ be a marking.
We call a step $MaM'$:

\begin{itemize}
\item \emph{e/l-0-persistent} iff it is \emph{e/e-persistent} (the execution of an action a does not disable any other action);
\item \emph{e/l-1-persistent} iff $(\forall b\in T, b\neq a) \ Mb \Rightarrow [Mab \lor (\exists c \in T) Macb]$ (the execution of an action a pushes the execution of any other enabled action away for at most 1 step);
\item \emph{e/l-2-persistent} iff $(\forall b\in T, b\neq a) \ Mb \Rightarrow (\exists w \in T^*) [|w|\leq 2 \land Mawb]$ (the execution of an action a pushes the execution of any other enabled action away for at most 2 steps);
\\
\indent\ldots

\item \emph{e/l-k-persistent} for some $\mathrm{k} \in \mathbb{N}$ iff $(\forall b\in T, b\neq a) \ Mb \Rightarrow (\exists w \in T^*) [{|w|\leq \mathrm{k}} \land Mawb]$ (the execution of an action a pushes the execution of any other enabled action away for at most $\mathrm{k}$ steps);
\\
\indent\ldots
\item \emph{e/l-$\infty$-persistent} iff $(\forall b\in T, b\neq a) \ Mb \Rightarrow (\exists w \in T^*) \ Mawb$ (the execution of an action a pushes the execution of any other enabled action away).
\end{itemize}
\end{definition}

\textbf{Remark:} Note that e/l-$\infty$-persistent steps are exactly e/l-persistent steps.
\\ \\
Directly from Definition \ref{d41} we get the
\begin{fact}
\label{f42}
Let $\mathrm{S}=(P,T,F,M_0)$ be a p/t-net, let $M$ be a marking. If the step $MaM'$ is e/l-$\mathrm{k}$-persistent for some $\mathrm{k} \in \mathbb{N}$, then it is also e/l-$\mathrm{i}$-persistent for every $\mathrm{i}\geq \mathrm{k}$.
\end{fact}

\begin{definition}
\label{d43}
Let $\mathrm{S}=(P,T,F,M_0)$ be a p/t-net, $M$ be a marking and $\mathrm{k}\in \mathbb{N}$.
Marking $M$ is \emph{e/l-k-persistent} iff for every action $a \in T$ that is enabled in $M$ the step $Ma$ is \emph{e/l-k-persistent}.
P/t-net $\mathrm{S}=(\mathrm{N},M_0)$ is \emph{e/l-k-persistent} iff every marking reachable in S is \emph{e/l-k-persistent}.
We denote the class of e/l-$\mathrm{k}$-persistent p/t-nets by $\mathrm{P}_{\mathrm{e/l-k}}$.
\end{definition}

\begin{example}
\label{e_f}
\begin{figure}[h]
\centering
\includegraphics[width=0.2\textwidth]{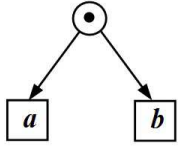}
\caption{A p/t-net (for Ex.\ref{e_f}) that is not e/l-k-persistent for any $\mathrm{k}\in \mathbb{N}$}
\label{Fig_F}
\end{figure}

Let us look at the example of Fig. \ref{Fig_F}. Both actions $a$ and $b$ are enabled in the initial marking. After the execution of the action $a$, the action $b$ is never enabled again, and after the execution of the action $b$, the action $a$ is never enabled again. So the net can not be e/l-k-persistent for any natural number k.
\end{example}

\begin{example}
Let us look at the example of Fig. \ref{Fig_E}. The net is not e/l-0-persistent but it is e/l-1-persistent.
\end{example}
\begin{example}
\label{e411}

\begin{figure}[h]
\centering
\includegraphics[width=0.4\textwidth]{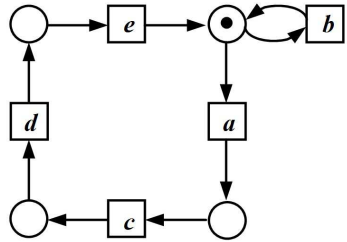}
\caption{A p/t-net (for Ex.\ref{e411}) that is e/l-3 persistent but not e/l-2 persistent}
\label{Fig2}
\end{figure}

Let us look at the example of Fig. \ref{Fig2}. The only possible situation for temporary disabling an action by another one is the execution of $a$ that disables $b$. And then $b$ could be enabled again after the execution of the sequence $cde$, so after 3 steps. Hence, the net is e/l-3-persistent, and obviously not e/l-2-persistent.
\end{example}
The direct conclusion from Fact \ref{f42} and Definition \ref{d43} is as follows:
\begin{fact}
\label{f44}
Let $\mathrm{S}=(P,T,F,M_0)$ be a p/t-net, $M$ be a marking, and $\mathrm{k}\in\mathbb{N}$.
If the marking $M$ is e/l-$\mathrm{k}$-persistent, then it is also e/l-$\mathrm{i}$-persistent for every $\mathrm{i}\geq \mathrm{k}$.
If the net S is e/l-$\mathrm{k}$-persistent, then it is also e/l-$\mathrm{i}$-persistent for every $\mathrm{i}\geq \mathrm{k}$.
\end{fact}
\mbox{ }\\
\textbf{Remark:} Based on this Fact \ref{f44} we can extend the existing hierarchy of persistent nets as shown in Figure \ref{FigB}.
\begin{figure}[h]
\centering
\includegraphics[width=0.8\textwidth]{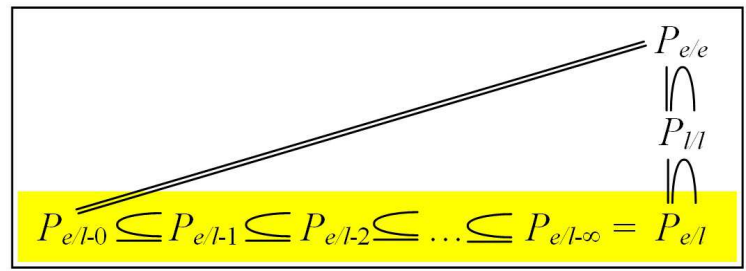}
\caption{A hierarchy of persistent nets - an extension}
\label{FigB}
\end{figure}

\subsection{Related decision problems}

\subsubsection{\textbullet\ EL-k Step Persistence Problem \\and EL-k Marking Persistence Problem}\mbox{ }\\ \\
Let $\mathrm{k}\in\mathbb{N}$ be a fixed natural number. Now we can formulate basic problems regarding the concept of e/l-k-persistence. \\ \\
\newpage

\mbox{ }\\
The first problem is as follows:
\\ \\
\textbf{EL-k Step Persistence Problem}
\\
\indent\textbf{\emph{Instance:}} P/t-net S, marking $M$, action $a \in T$ enabled in $M$.\\
\indent\textbf{\emph{Question:}}Is the step $Ma$ e/l-k-persistent?

\begin{theorem}
\label{t45}
The EL-k Step Persistence Problem is decidable (for any $\mathrm{k} \in \mathbb{N}$).
\end{theorem}

\begin{proof}
An algorithm to check if a step $Ma$ is e/l-k-persistent (for some $\mathrm{k} \in \mathbb{N}$) for a given net $\mathrm{S}=(\mathrm{N},M_0)$:\\
Let us build the part of the depth of k+1 (we call it the (k+1)-component) of the reachability tree of $(\mathrm{N},M')$, where $M'$ is a marking obtained from $M$ by execution of $a$. The step $Ma$ is e/l-k-persistent if for every action $b \in T$, such that $a\neq b$ and $b$ is enabled in $M$, there is a path in the (k+1)-component of the reachability tree of $(\mathrm{N},M')$ containing an arc labeled by $b$. 	
\end{proof}\mbox{ }\\
Let us introduce another problem:
\\ \\
\textbf{EL-k Marking Persistence Problem}
\\
\indent\textbf{\emph{Instance:}}P/t-net $\mathrm{S}=(\mathrm{N}, M_0)$, marking $M$.\\
\indent\textbf{\emph{Question:}}Is the marking $M$ e/l-k-persistent?

\begin{theorem}
\label{t46}
The EL-k Marking Persistence Problem is decidable \\
\indent\indent\indent (for any $\mathrm{k} \in \mathbb{N})$.
\end{theorem}

\begin{proof}
For every action $a\in T$ that is enabled in a marking $M$, we check if a step $Ma$ is e/l-k-persistent (for some $\mathrm{k} \in \mathbb{N}$) for a given net $S=(\mathrm{N}, M_0)$, using the algorithm of Theorem \ref{t45}. 	
\end{proof}

\subsubsection{\textbullet \ EL-k Net Persistence Problem}\mbox{ }\\ 
\\
Let us consider the following problem:
\\ \\
\textbf{EL-k Net Persistence Problem}
\\
\indent\textbf{\emph{Instance:}}P/t-net $\mathrm{S}=(\mathrm{N}, M_0),\mathrm{k}\in\mathbb{N}$.\\
\indent\textbf{\emph{Question:}}Is the net S e/l-k-persistent?
\\
\\To solve this problem we must prove a set of auxiliary facts.
\\ \\
From this moment, let $\mathrm{S}=(\mathrm{N}, M_0)$ be an arbitrary p/t-net.
\\ \\
Let us define the following set of markings:\\
$\mathrm{E}_{a,b}=\{M \in \mathbb{N}^{P} \ | \  Ma\land Mb\}$- the set of markings enabling actions $a$ and $b$ simultaneously. \\ 
\\
Let us define $\mathrm{minE}_{a,b} \in \mathbb{N}^{P}$, the minimum marking enabling actions $a$ and $b$ simultaneously: if $(^\bullet a[\mathrm{i}]=1 \lor \ ^\bullet b[\mathrm{i}]=1)$ then $\mathrm{minE}_{a,b}[\mathrm{i}]:=1$ else $\mathrm{minE}_{a,b}[\mathrm{i}]:=0$ (for $\mathrm{i}=\{1,\ldots,|P|\})$.\\
Note that $\mathrm{E}_{a,b}=\mathrm{minE}_{a,b}+ \mathbb{N}^{P}$.

\subsubsection{\textbullet \ Mutual Enabledness Reachability Problem}\mbox{ }\\
\\
Let us formulate an auxiliary problem:
\\ \\
\textbf{Mutual Enabledness Reachability Problem}
\\
\indent\textbf{\emph{Instance:}}P/t-net $\mathrm{S}=(\mathrm{N}, M_0)$, actions $a,b\in T$.\\
\indent\textbf{\emph{Question:}}Is there a marking $M$ such that $M\in \mathrm{E}_{a,b}$ and $M \in[M_0\rangle$ ? \\
\indent(Is there a reachable marking $M$ such that \\ \indent actions $a$ and $b$ are both enabled in $M$?)
\\

\begin{theorem}
\label{t413}
The Mutual Enabledness Reachability Problem is decidable.
\end{theorem}

\begin{proof}
Let $M=\mathrm{minE}_{a,b}$. We build a coverability graph for the p/t-net S. We check whether in the graph exists a vertex corresponding to an $\omega$-marking $M'$ such that $M'$ covers $M$. If so, then actions $a$ and $b$ are simultaneously enabled in some reachable marking of the net S. Otherwise, those transitions are never enabled at the same time. 	
\end{proof}
Let $\mathrm{Min}[M_0\rangle$ be the set of minimal (wrt $\leq$) reachable markings of the net S. As members of $\mathrm{Min}[M_0\rangle$ are incomparable, the set $\mathrm{Min}[M_0\rangle$  is finite, by Lemma \ref{l26}.\\ \\
Le us denote by $\mathrm{RE}_{a,b}$ the set of all reachable markings of the net S enabling actions $a$ and $b$ simultaneously: $\mathrm{RE}_{a,b}= \{M \in [M_0\rangle \ | \ Ma\land Mb\} = \mathrm{E}_{a,b} \cap  [M_0\rangle$. \\ \\
Let  $\mathrm{Min}(\mathrm{RE}_{a,b})$ be a set of all minimal reachable markings of the net S enabling action $a$ and $b$ simultaneously.

\subsubsection{\textbullet \ Results of Valk and Jantzen}\mbox{ }\\
\\
In order to construct the set $\mathrm{Min}(\mathrm{RE}_{a,b})$, we put into work the theory of residue sets of Valk/Jantzen \cite{ValkJantzen}.

\begin{definition}[Valk/Jantzen \cite{ValkJantzen}]
A subset $X\subseteq \mathbb{N}^\mathrm{k}$ has property $\mathrm{RES}$  if and only if the problem "Does  $\downarrow v$ intersect $X$?" is decidable for any  $\omega$-vector $v\in \mathbb{N}_\omega^\mathrm{k}$.
\end{definition}

\begin{theorem}[Valk/Jantzen \cite{ValkJantzen}]
\label{twValk}
Let $X\subseteq \mathbb{N}^\mathrm{k}$ be a right-closed set. Then the set $Min(X)$ is effectively computable if and only if $X$ has property $\mathrm{RES}$.
\end{theorem}

\subsubsection{\textbullet \ Set Reachability Problem}\mbox{ }\\
\\
We also use the fact of decidability of the Set Reachability Problem for rational convex sets (Def. \ref{def_mon},\ref{def_clos}), proved in  \cite{BarOch}.
\\ \\
\textbf{Set Reachability Problem}
\\
\indent\textbf{\emph{Instance:}}P/t-net $\mathrm{S}=(\mathrm{N}, M_0)$ and a set $X\subseteq\mathbb{N}^{P}$.\\
\indent\textbf{\emph{Question:}}Is there a marking $M\in X$, reachable in S? \\

\begin{theorem}
\label{tplus}
The Set Reachability Problem is decidable for rational convex sets in p/t-nets.
\end{theorem}
The Set Reachability Problem is a generalization of the classical Marking Reachability Problem. The proof uses decidability of the Reachability Problem.
\subsubsection{\textbullet \ Minimal reachable markings enabling two actions simultaneously}\mbox{ }\\ 
\\
Now we are ready to prove:

\begin{proposition}
\label {p414}
The set $\mathrm{Min}(\mathrm{RE}_{a,b})$ can be effectively constructed for a given net  $\mathrm{S}=(P,T,F,M_0)$.
\end{proposition}

\begin{proof}
Let us take the right closure $\mathrm{RE}_{a,b}\uparrow$ of the set $\mathrm{RE}_{a,b}$.
\\Note that $\mathrm{Min}(\mathrm{RE}_{a,b})=\mathrm{Min}(\mathrm{RE}_{a,b}\uparrow)$. To show that the set of minimal elements of the set $\mathrm{RE}_{a,b}$ is effectively computable, it is enough to demonstrate that the set $\mathrm{RE}_{a,b}\uparrow$  has the property RES (i.e. for any  $\omega$-vector $v\in \mathbb{N}_\omega^{P}$ the problem "$(\downarrow  v\cap RE_{a,b}\uparrow \neq \emptyset)$?" is decidable) and apply Theorem \ref{twValk}. \\
Let $X=\downarrow v\cap E_{a,b}$, where $E_{a,b}=\mathrm{min}E_{a,b}+ \mathbb{N}^{P}$. Let us notice, that  $\downarrow v$ is a convex set, hence rational (Proposition \ref{p26}). The set $E_{a,b}$ is also a rational convex set. As an intersection of convex rational sets, the set $X$ is convex and rational (Theorem \ref{twGins}) as well.\\
Hence, putting into work decidability of the Set Reachability Problem for rational convex sets (Theorem \ref{tplus}) we decide whether any marking from the set $X$ is reachable in S. Therefore, we can decide whether the set $X'= \downarrow v\cap \mathrm{RE}_{a,b}$ is nonempty. (It is the case when at least one marking from the set $X$ is reachable in S.) Let us notice that the set $X''=\downarrow v\cap \mathrm{RE}_{a,b}\uparrow$  is nonempty if and only if the set $X'$ is nonempty. That is why the set $\mathrm{RE}_{a,b}\uparrow$  has the property RES, and consequently the set $\mathrm{Min}(\mathrm{RE}_{a,b})$ is effectively computable by Theorem \ref{twValk}. 	
\end{proof}

\begin{example}
\label{e415}

\begin{figure}[h]
\centering
\includegraphics[width=0.4\textwidth]{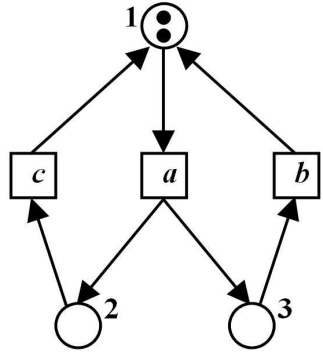}
\caption{A p/t-net for Ex.\ref{e415}.}
\label{Fig4}
\end{figure}

The set of all minimal reachable markings of the net depicted in Figure \ref{Fig4} enabling action $a$ and $b$ simultaneously, is $\mathrm{Min}(\mathrm{RE}_{a,b})=\{[1,1,1], [2,0,1]\}$.
\end{example}

\begin{proposition}
\label{p416}
If there exists a marking  $M \in\mathrm{RE}_{a,b}$ such that the execution of an action $a$ in $M$ pushes the execution of an action $b$ away for more than $\mathrm{k}$ steps (for some $\mathrm{k}\in\mathbb{N}$), then there exists some minimal marking  $M'\in\mathrm{Min}(\mathrm{RE}_{a,b})$ such that the execution of an action $a$ in $M'$ pushes the execution of an action $b$ away for more than $\mathrm{k}$ steps, too.
\end{proposition}

\begin{proof}
Let $M$ be a marking, such that the execution of an action $a$ in $M$ pushes the execution of an action $b$ away for more than k steps (for some $\mathrm{k}\in\mathbb{N}$). Let $M'\in  \mathrm{Min}(\mathrm{RE}_{a,b})$ such that $M'\leq M$. Such a marking has to exist. Suppose that there is a string $w\in T^*$, $|w|\leq \mathrm{k}$ such that $M'awb$. Then obviously also $Mawb$ (from the monotonicity property - Fact \ref{f27}). We obtain a contradiction. Hence, the execution of an action $a$ in $M'$ postpones the execution of $b$ for more than k steps. 	
\end{proof}

\subsubsection{\textbullet \ EL-k Transition Persistence Problem and EL-k Net Persistence Problem}\mbox{ }\\
\\
Now, we are ready to introduce the following problem:
\\ \\
\textbf{EL-k Transition Persistence Problem}
\\
\indent\textbf{\emph{Instance:}}P/t-net $\mathrm{S}=(\mathrm{N}, M_0)$, ordered pair $(a,b)\in T\times T, b\neq a$, $\mathrm{k}\in\mathbb{N}$.\\
\indent\textbf{\emph{Question:}}Is there a reachable marking $M \in [M_0\rangle$ such that  \\
\indent\indent \indent\indent$Ma \land Mb \land \lnot[(\exists w\in T^*) |w|\leq \mathrm{k} \land Mawb]$?
\\
\indent\indent \indent\indent(Does $a$ postpone $b$ for more than k steps?)

\begin{theorem}
\label{t417}
The EL-k Transition Persistence Problem is decidable.
\end{theorem}

\begin{proof}
We introduce an algorithm of deciding if an action $a$ pushes the execution of an action $b$ away for more than k steps in some reachable marking $M$.
\begin{enumerate}
\item We check whether both actions $a$ and $b$ are enabled in some reachable marking (using decidability of Mutual Enabledness Reachability Problem).
\begin{enumerate}
\item If not, we answer NO.
\item Otherwise:
\begin{enumerate}
\item We build the set $\mathrm{Min}(\mathrm{RE}_{a,b})$. This set can be effectively computed by Proposition \ref{p414} using Valk/Jantzen algorithm.
\item For all markings $M_1 \in \mathrm{Min}(\mathrm{RE}_{a,b})$:\\
$M_2:=M_1a$. \\
We build an initial part of the depth of k+1 (the (k+1)-component) of the reachability tree of $(\mathrm{N}, M_2)$. If the piece has an edge labeled by $b$, we answer NO. Otherwise we answer YES. 	
\end{enumerate}
\end{enumerate}
\end{enumerate}
\end{proof}
And now the proof of decidability of the EL-k Net Persistence Problem is ready.

\begin{theorem}
\label{t418}
The EL-k Net Persistence Problem is decidable (for any $\mathrm{k}\in \mathbb{N}$).
\end{theorem}

\begin{proof}
S is e/l-k-persistent iff the algorithm solving EL-k Transition Persistence Problem answers NO for all ordered pairs $(a,b)\in T\times T$, $a\neq b$.
\end{proof}
\vspace*{-1cm}
\begin{example}
\label{e_fig_g}
\begin{figure}[h]
\centering
\includegraphics[width=0.4\textwidth]{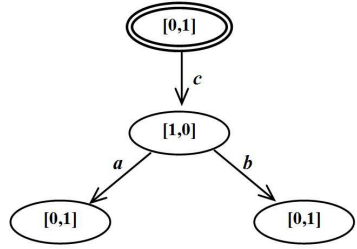}
\caption{2-component of the reachability tree of the net of Figure \ref{Fig_E}.}
\label{Fig_G}
\end{figure}

Let us check whether the action $a$ of the net S of Figure \ref{Fig_E} postpones the action $b$ for more than 1 step.\\
Actions $a$ and $b$ are both enabled in the initial marking. \\
The set $\mathrm{Min}(\mathrm{RE}_{a,b})$ consists of a single marking $M_1=[1,0]$. 
We take $M_2=M_1a=[0,1]$. We build a 2-component of the reachability tree of the net $(\mathrm{S},M_2)$. The tree is depicted in Figure \ref{Fig_G}. The tree has an edge labeled by $b$ so the action $a$ does not postpone the action $b$ for more than 1 step.
\end{example}

\subsubsection{\textbullet \ EL-k Transition Persistence Problem - an alternative approach}\mbox{ }\\
\\
In order to show decidability of the EL-k Net Persistence Problem we can use the technique used for proving decidability of LL Net Persistence Problem and EL Net Persistence Problem, presented in \cite{BarOch}.
\\ \\
Again, we deal with the EL-k Transition Persistence Problem, crucial for the proof. We show an alternative proof of decidability of that problem.

\mbox{ }\\
\textbf{EL-k Transition Persistence Problem}
\\
\indent\textbf{\emph{Instance:}}P/t-net $\mathrm{S}=(\mathrm{N}, M_0)$, ordered pair $(a,b)\in T\times T, b\neq a$, $\mathrm{k}\in\mathbb{N}$.\\
\indent\textbf{\emph{Question:}}Is there a reachable marking $M \in [M_0\rangle$ such that  \\
\indent\indent \indent\indent$Ma \land Mb \land \lnot[(\exists w\in T^*) |w|\leq \mathrm{k} \land Mawb]$?
\\
\\
Let us define, in order to reformulate the problem above, the following sets of markings:\\
$\mathrm{E}_a = \{M\in \mathbb{N}^{P} \ | \ Ma\}$ - markings enabling $a$ \\
$\mathrm{E}_b = \{M\in \mathbb{N}^{P} \ | \ Mb\}$ - markings enabling $b$ \\
$\mathrm{E}_{a(k)b} = \{M\in \mathbb{N}^{P} \ | (\exists w\in T^*) |w|\leq\mathrm{k} \land Mawb\}$ - markings enabling a such that after the execution of a the action b is potentially enabled after at most k steps.
\\
\\
Now we can reformulate the question of the  problem above:
\\
\indent\textbf{\emph{Question:}}Is the set
$\mathrm{E}_a\cap \mathrm{E}_b \cap (\mathbb{N}^{P}-\mathrm{E}_{a(k)b})$ reachable in $(\mathrm{N}, M_0)$?
\\
\\
Let us look again at\\ \\
\textbf{Theorem \ref{t418}}. 
\textit{The EL-k Net Persistence Problem is decidable (for any $\mathrm{k}\in \mathbb{N}$).}

\begin{proof}
First note that, by the monotonicity property (Fact \ref{f27}), the set \newline$\mathrm{E}_a\cap \mathrm{E}_b\cap (\mathbb{N}^{|\mathrm{P}|}-\mathrm{E}_{a(k)b})$ is convex, thus rational (by Proposition \ref{p26}). The rational expressions for $E_a$ and $E_b$ are  $E_a=^\bullet a+\mathbb{N}^\mathrm{k}$ and $E_b=^\bullet b+\mathbb{N}^\mathrm{k}$. Clearly, the set $\mathrm{E}_{a(k)b}$ is right-closed, by the monotonicity property. We shall prove that it has the property RES. Namely, $\downarrow\!v$ ($v\in\mathbb{N}_\omega^{P}$)intersects $\mathrm{E}_{a(k)b}$ if and only if $^\bullet a\leq v$ (i.e. $a$ is enabled in $v$) and there is a path in the reachability tree, limited to (k+1) first levels, of the net $(\mathrm{N}, v')$, where $v'$ is an $\omega$ -marking obtained from $v$ by the execution of $a$, containing an arc labelled by $b$. It is obviously decidable. Hence, the set $\mathrm{E}_{a(k)b}$ has the property RES, thus (by Theorem \ref{twValk}) the set $\mathrm{Min}(\mathrm{E}_{a(k)b})$ is effectively computable. As $\mathrm{E}_{a(k)b}$ is right-closed, we get the rational expression for it: $\mathrm{E}_{a(k)b} = \mathrm{Min}(\mathrm{E}_{a(k)b})+\mathbb{N}^{P}$. Finally, using Theorem \ref{twGins} of Ginsburg/Spanier \cite{GinsburgSpanier}, we compute rational expression for $\mathrm{E}_a\cap \mathrm{E}_b \cap (\mathbb{N}^{P}-\mathrm{E}_{a(k)b})$ and Theorem \ref{twValk} yields decidability of the problem. 	
\end{proof}

\subsection{Collapsing of the hierarchy of e/l-persistence}

\subsubsection{\textbullet \ k-enabledness}\mbox{ }\\
\\
Let us recall the well-known fact, that follows from the Dickson's Lemma (Lemma~\ref{l26}).

\begin{fact}
\label{f47}
Every infinite sequence of markings contains an infinite increasing (not necessarily strictly) subsequence of markings.
\end{fact}
Recall also that p/t-nets have the monotonicity property - Fact \ref{f27}.

\mbox{ }\\
Let us define the notion of k-enabledness.
\begin{definition}[k-enabledness]
\label{d48}
Let $\mathrm{S}=(P,T,F,M_0)$ be a p/t-net, let $M$ be a marking. For $\mathrm{k} \in \mathbb{N}$ we say that the action $a\in T$ is \emph{k-enabled} in the marking $M$ if and only if $\exists w \in T^*$, such that $|w|\leq \mathrm{k} \land Mwa$.
\end{definition}
\newpage
Now, we can show:

\begin{lemma}
\label{l49}
Let $\mathrm{S}$ be a p/t-net. For an arbitrary $a \in T$ there exists a natural number $\mathrm{k}_a\in \mathbb{N}$, such that in every marking $M$ the transition $a$ is $\mathrm{k}_a$-enabled or it is dead.
\end{lemma}
\begin{proof}
Suppose that the lemma does not hold for some action $a\in T$. It means that for each $\mathrm{k}\in\mathbb{N}$ there is a marking $M$ such that $a$ is not k-enabled but not dead. This means that $a$ is $\mathrm{k}'$-enabled for some $\mathrm{k}'>\mathrm{k}$. Thus, there exists  an infinite sets of markings $M_1, M_2,\ldots$ and integers $\mathrm{k}_1<\mathrm{k}_2<\ldots$, such that the action $a$ is live in each marking $M_\mathrm{i}$ and it is not $\mathrm{k}_\mathrm{i}$-enabled in $M_\mathrm{i}$ for all $\mathrm{i}=1,2,\ldots$.
Let us choose (by Fact \ref{f47}) an infinite increasing sequence of markings $M_{\mathrm{i}1}\leq M_{\mathrm{i}2}\leq \ldots$.
Since the action $a$ is live in $M_{\mathrm{i}1}$, it is k-enabled in $M_{\mathrm{i}1}$, for some $\mathrm{k}\in\mathbb{N}$. As the strictly increasing sequence $\mathrm{k}_1<\mathrm{k}_2<\ldots$ is infinite, $\mathrm{k}<\mathrm{k}_{\mathrm{ij}}$ for some j. By the monotonicity property (Fact \ref{f27}), the action $a$ is k-enabled, hence  $\mathrm{k}_{\mathrm{ij}}$-enabled in the marking $M_{\mathrm{ij}}$. Contradiction.	
\end{proof}
\mbox{ }\\
\textbf{Remark:} Note that the proof of Lemma \ref{l49} is purely existential, it does not present any algorithm for finding k.
\\
\\
Now, we are ready to formulate the main theorem of the section:

\begin{theorem}
\label{t410}
If a p/t-net is e/l-persistent, then it is e/l-$\mathrm{k}$-persistent for some $\mathrm{k} \in \mathbb{N}$.\\
In words: Whenever an action is disabled by another one, it is pushed away for not more than $\mathrm{k}$-steps.
\end{theorem}

\begin{proof}
If the net is e/l-persistent, then no action kills another enabled one. From the Lemma \ref{l49} we know, that if an action $a \in T$ is not dead then it is $\mathrm{k}_a$-enabled.
Let us take $\mathrm{K}=\mathrm{max}\{\mathrm{k}_a | a\in T\}$, for the numbers $\mathrm{k}_a$ from the Lemma \ref{l49}. One can see that every action in the net that is not dead, is K-enabled. Thus, the execution of any action may postpone the execution of an action $a$ for at most K steps. 	
So we have the implication: if a p/t-net is e/l-persistent, then it is e/l-K-persistent, for K defined above. 	
\end{proof}
\mbox{ }\\
\textbf{Remark:} As the proof of Lemma \ref{l49} explicitly uses the monotonicity property of p/t-nets, the Theorem \ref{t410} holds only for nets satisfying this property. The following example shows that Theorem \ref{t410} does not hold for nets without the monotonicity property (for instance, inhibitor nets).

\begin{example}
\label{e412}

\begin{figure}[h]
\centering
\includegraphics[width=0.5\textwidth]{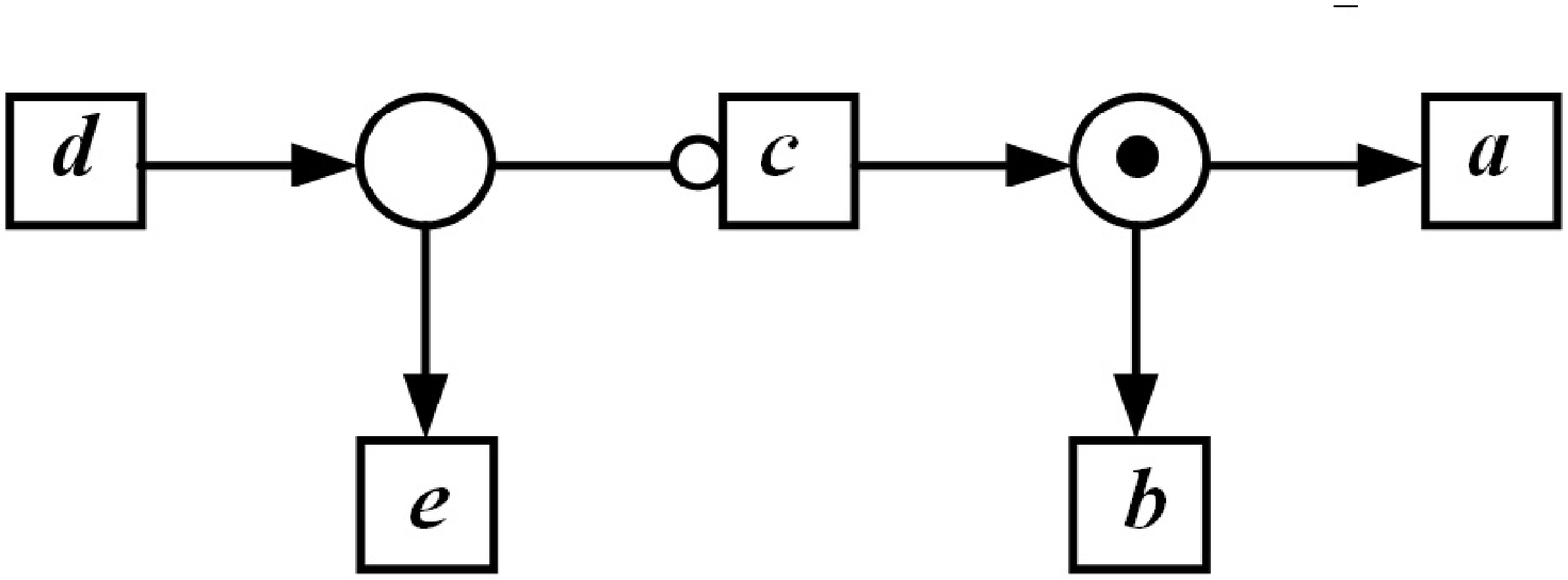}
\caption{An inhibitor net for Ex.\ref{e412}}
\label{Fig3}
\end{figure}

Let us look at the example of Fig. \ref{Fig3}. We can see an inhibitor net and its computation such that for every $\mathrm{k}\in\mathbb{N}$ one can push an action away for a distance greater than k steps.\\
This net is live, hence it is e/l-persistent, but it is not e/l-k-persistent for any $\mathrm{k}\in\mathbb{N}$.\\
In the infinite computation $acbcdaecbcddaeecbcdddaeeecb\ldots$ the first $a$ pushes $b$ away for 1 step, the second - for 2 steps and every k-th $a$ - for k steps.
\end{example}

\subsubsection{\textbullet \ Collapsing of the hierarchy - an effective proof}\mbox{ }\\
\\
Finally, let us recall other decision results of \cite{BarOch}:
\\ \\
\textbf{Transitions Persistence Problems}
\\
\indent\textbf{\emph{Instance:}}P/t-net $\mathrm{S}=(\mathrm{N}, M_0)$, and transitions $a,b\in T, a\neq b$.\\
\indent\textbf{\emph{Questions (informally):}}\\
\indent \indent\indent EE-Persistence Problem: Does $a$ disable an enabled $b$?\\
\indent \indent\indent LL-Persistence Problem: Does $a$ kill a live $b$?\\
\indent \indent \indent EL-Persistence Problem: Does $a$ kill an enabled $b$?\\
\\
From \cite{BarOch} we know that the problems are decidable.

\begin{theorem}
\label{t419}
For a given p/t-net $\mathrm{S}=(\mathrm{N},M_0)$ and a pair of transitions $a,b\in T$ one can calculate a minimum number $\mathrm{k}_{a,b}\in \mathbb{N}$ such that $a$ postpones an enabled~$b$ for at most $\mathrm{k}_{a,b}$ steps (if such a number exists).
\end{theorem}

\begin{proof}\
\begin{itemize}
\item We check whether both actions $a$ and $b$ are enabled in some reachable marking (using decidability of Mutual Enabledness Reachability Problem). If not, $\mathrm{k}_{a,b}$ does not exist (actions $a$ and $b$ are never enabled at the same time). \
Otherwise:
\begin{itemize}
\item We ask whether $a$ kills an enabled $b$ (EL-Persistence Problem).\\
If YES then $\mathrm{k}_{a,b}$ does not exist ($a$ kills $b$)\\
else:\
\begin{itemize}
\item We compute the set $\mathrm{Min}(\mathrm{RE}_{a,b})$.This set can be effectively computed by Proposition \ref{p414} using Valk/Jantzen algorithm.\
\item We build the initial part of reachability tree of the net $S$ as long as from every $M\in \mathrm{Min}(\mathrm{RE}_{a,b})$ we get a marking $M'$ with the property that a path leads to a~vertex $M'$ (it can be an empty path) such that $M'b$. Clearly, such part of the tree is finite, as we get the whole $\mathrm{Min}(\mathrm{RE}_{a,b})$ and for any $M\in \mathrm{Min}(\mathrm{RE}_{a,b})$ a finite path leading from $M$ to a vertex $M'$ such that $M'b$. The maximum length of such paths is the desired number $\mathrm{k}_{a,b}$. 	
\end{itemize}
\end{itemize}
\end{itemize}
\end{proof}

\begin{theorem}
\label{t420}
If a p/t-net $\mathrm{S}=(\mathrm{N},M_0)$ is e/l-persistent, then it is e/l-$\mathrm{k}$-persistent for some $\mathrm{k}\in \mathbb{N}$ and such a $\mathrm{k}$ can be effectively computed.
\end{theorem}

\begin{proof}
For every pair $(a,b)$ of transitions we find $\mathrm{k}_{a,b}$ defined above. The number we are looking for is $\mathrm{k}=\mathrm{max}(\mathrm{k}_{a,b}: a,b\in T)$. 	
\end{proof}
\mbox{ }\\ \\
We established that an action can not postpone another action (without killing it) indefinitely (Theorem \ref{t410}). We proved, that if a p/t-net is e/l-persistent, then it is e/l-k-persistent for some $\mathrm{k}\in \mathbb{N}$. We showed that such a k exists and we present any algorithm for finding this k. \\ \\

\section{Conclusions}

It is shown in \cite{BarMikOch} that if we change the firing rule in the following way: only e/e-persistent computations are permitted, then we get a new class of nets (we call them \emph{nonviolence nets}) which are computationally equivalent to Turing machines.
We plan to investigate net classes, with firing rules changed (only e/l-k-persistent computations are allowed) and answer the question:
\\ \\
\textbf{Question 1:}\\
What is the computational power of nets created this way?
\\ \\
In this paper, we have investigated the hierarchy of persistence in p/t-nets. We would like to study the hierarchy of e/l-k-persistence in some extensions of p/t-nets, for instance nets with read arcs and reset nets. 
All results of the paper hold for nets with read arcs (\cite{MontanariRossi}), as they can be simulated by classical Petri nets with self-loops with the same reachability set (but with distinct step semantics).
On the contrary, only Lemma \ref{l49} and Theorem \ref{t410} hold (with the same proof) for other extended Petri nets posessing the monotonicity property (e.g. reset, double, transfer nets), but the results supported with the fact of decidability of the Reachability Problem (Proposition \ref{p414}, Theorem \ref{t418}, Theorem \ref{t419}) cannot be applied to those nets, because of undecidability of the Reachability Problem in them (see \cite{Dufourd}).

\bibliography{TCS}{}
\bibliographystyle{plain}

\end{document}